\documentclass[groupedaddress, aps, longbibliography, reprint, prx]{revtex4-1}

\usepackage{graphicx}
\usepackage{bm}
\usepackage{bbold}
\usepackage{physics}
\usepackage{xcolor}
\usepackage{enumitem}
\usepackage{amsmath, amssymb, amsthm}
\usepackage[normalem]{ulem}
\usepackage{natbib}
\usepackage{comment}

\newtheorem{theorem}{Theorem}

\newtheorem{fact}{Fact}[section]
\newtheorem{lemma}{Lemma}

\newtheorem{definition}{Definition}

\usepackage[colorlinks=true]{hyperref}
\hypersetup{citecolor = blue, linkcolor=blue, urlcolor=blue} 

\newcommand{\Wg}{\mathsf{Wg}}
\renewcommand{\(}{\left(}
\renewcommand{\)}{\right)}
\renewcommand{\[}{\left[}
\renewcommand{\]}{\right]}

\newcommand{\tracenorm}[1]{\left \| #1  \right \|_{\rm tr}}

\newcommand{\opnorm}[1]{\left \| #1  \right \|_{\rm op}}

\begin{document}

\title{Hardness of observing strong-to-weak symmetry breaking}


\author{Xiaozhou Feng}
\thanks{These two authors contributed equally.} \affiliation{Department of Physics, The University of Texas at Austin, Austin, TX 78712, USA}

\author{Zihan Cheng}
\thanks{These two authors contributed equally.}
\affiliation{Department of Physics, The University of Texas at Austin, Austin, TX 78712, USA}

\author{Matteo Ippoliti}
\affiliation{Department of Physics, The University of Texas at Austin, Austin, TX 78712, USA}

\begin{abstract}
Spontaneous symmetry breaking (SSB) is the cornerstone of our understanding of quantum phases of matter. Recent works have generalized this concept to the domain of mixed states in open quantum systems, where symmetries can be realized in two distinct ways dubbed {\it strong} and {\it weak}. Novel intrinsically mixed phases of quantum matter can then be defined by the spontaneous breaking of strong symmetry down to weak symmetry. However, proposed order parameters for strong-to-weak SSB (based on mixed-state fidelities or purities) seem to require exponentially many copies of the state, raising the question: is it possible to {\it efficiently} detect strong-to-weak SSB in general? Here we answer this question negatively in the paradigmatic cases of $\mathbb{Z}_2$ and $U(1)$ symmetries. We construct ensembles of pseudorandom mixed states that do not break the strong symmetry, yet are computationally indistinguishable from states that do. This rules out the existence of efficient state-agnostic protocols to detect strong-to-weak SSB. 
\end{abstract}

\maketitle

{\it Introduction.}---Quantum phases of matter are best understood for ground states of isolated quantum systems, described by pure states~\cite{Wen_zoo_2017}. Their robustness to perturbations that cause the state to become mixed (e.g. finite temperature, decoherence) is an interesting fundamental question~\cite{Dennis_topological_2002,Hastings_topological_2011,Lu_detecting_2020,Fan_diagonstics_2024}. Recent works have pointed out another interesting possibility: the existence of {\it intrinsically mixed} quantum phases of matter, exhibiting patterns of symmetry breaking without a pure-state analogue~\cite{Lee_Quantum_2023,Lee_Symmetryprotected_2025,Lessa_strong_2025,Wang_Intrinsic_2025,Ellison_Toward_2025,Sohal_Noisy_2025,Zhang_strong_2025,Chen_separability_2024, Ma_topologicalphasesaveragesymmetries_2025, Ma_symmetry_2025}.

This possibility arises because mixed states can realize a symmetry in two physically distinct ways: an {\it exact} or {\it strong} symmetry, and an {\it average} or {\it weak} symmetry. 
Precisely, for a unitary $U$ representing the action of a given symmetry on the Hilbert space, a mixed state $\rho$ has a strong symmetry if  $U\rho \propto\rho$, and a weak symmetry if $U\rho U^\dagger = \rho$. 
Both conditions require that $\rho$ can be unraveled into a mixture of symmetric pure states, however strong symmetry also requires all states in the mixture to have the same symmetry charge. 
This richer structure of symmetries enables a route to SSB that is unique to mixed states, known as strong-to-weak spontaneous symmetry breaking (SWSSB). Many intriguing properties of this phenomenon have been recently investigated in quantum systems at finite temperature or under the effect of local decoherence channels~\cite{Lee_Quantum_2023,Lee_Symmetryprotected_2025,Lessa_strong_2025,Wang_Intrinsic_2025,Ellison_Toward_2025,Sohal_Noisy_2025,Zhang_strong_2025,Chen_separability_2024, Ma_topologicalphasesaveragesymmetries_2025, Ma_symmetry_2025,Kuno_strong_2024,Gu_spontaneoussymmetrybreakingopen_2024,zhang_fluctuation_2024,Huang_Hydrodynamics_2025}.

In pure states, SSB can be characterized by long range order in the correlation function $\langle O_iO_j^\dagger\rangle$, where $O_i$ and $O_j$ are local operators that are charged under the global symmetry~\cite{Sachdev_2011}. 
In contrast, known diagnostics of SWSSB do not take the form of local expectation values; they are instead information-theoretic functions that are nonlinear in the density matrix $\rho$. 
Such diagnostics include the R\'enyi-2 correlator~\cite{Lee_Quantum_2023,Ma_symmetry_2025,Sala_spontaneous_2024}, the fidelity correlator~\cite{Lessa_strong_2025,zhang_fluctuation_2024}, and the R\'enyi-1 (or Wightman) correlator~\cite{weinstein_efficient_2024,Liu_diagnosingstrongtoweaksymmetrybreaking_2024}. 
These present different advantages and disadvantages in terms of analytical tractability, robustness to strongly-symmetric perturbations, etc, but they share a common limitation in terms of experimental access. Indeed, while nonlinearity in $\rho$ is not an issue in and of itself (it can be handled by protocols such as classical shadows~\cite{huang_predicting_2020,elben_randomized_2023,sun2025schemedetectstrongtoweaksymmetry}), each proposed diagnostic involves the measurement of quantities like mixed-state fidelity and purity which cannot be carried out efficiently in general. Thus, despite specific recent proposals \cite{weinstein_efficient_2024,sun2025schemedetectstrongtoweaksymmetry}, the experimental observability of SWSSB in general remains an open question.

In this work we conclusively answer this question. We establish that, given copies of a strongly-symmetric density matrix $\rho$ and no additional information, {\it no efficient protocol} can decide whether or not $\rho$ exhibits SWSSB in general. 
We achieve this result by leveraging recent constructions of pseudorandom unitaries (PRUs)~\cite{ji_pseudorandom_2018,metger_simple_2024,ma2024constructrandomunitaries,chen_pseudorandom_2024} and pseudorandom density matrices~\cite{Bansal_pseudorandom_2025}. Specifically, we construct ensembles of mixed states that do not exhibit SWSSB, yet are computationally indistinguishable from states that do. ``Computationally indistinguishable'' means that the two ensembles cannot be distinguished by any protocol using resources (time, number of state copies) scaling polynomially in the number of qubits $N$. Since an efficient state-agnostic protocol to detect SWSSB would be able to distinguish the two ensembles, such a protocol cannot exist. 
Our work thus establishes a fundamental limitation on the observability of intrinsicallly mixed phases of matter.

In the following we present the general idea of ``pseudo-SWSSB'' state ensembles and their construction for the paradigmatic cases of $\mathbb{Z}_2$ and $U(1)$ symmetries. 

{\it Pseudo-SWSSB.}---SSB in pure states is diagnosed by long-ranged correlations $C(i,j) = {\rm Tr}(\rho O_i O_j^\dagger)$, with $O_i$ a local operator charged under the symmetry group $G$. Such correlations are efficiently measurable in experiment. In contrast, SWSSB in mixed states is described by long-ranged correlations of different types, which may not be easily accessed experimentally. We review three commonly used diagnostics.
The R\'enyi-2 correlator, 
$R_2(i,j):={\Tr\(\rho O_iO_j^\dagger\rho O_i^\dagger O_j\)} / {\Tr\rho^2}$, 
requires precise measurement of the purity $\Tr\rho^2$, which can be exponentially small leading to an exponentially large sample complexity. 
The fidelity correlator, defined by
$F(i,j):=\Tr\sqrt{\sqrt{\rho}O_iO_j^\dagger\rho O_i^\dagger O_j\sqrt{\rho}}$,
is the fidelity between the two mixed states $\rho$ and $O_iO_j^\dagger\rho O_i^\dagger O_j$, whose measurement in general requires exponentially-costly tomography~\cite{Yuen_improvedsample_2023} (recent work presents a series expansion in terms of R\'enyi-$2n$ correlators~\cite{zhang_fluctuation_2024},  which is at least as hard as measuring $R_2$).
Finally, the R\'enyi-1 or Wightman correlator $R_1(i,j):=\Tr\(\sqrt{\rho} O_iO_j^\dagger\sqrt{\rho} O_i^\dagger O_j\)$
can be efficiently measured with knowledge of how to prepare the ``canonical purification'' state $\sqrt{\rho}$ \cite{weinstein_efficient_2024}, but in the absence of this information, requires state tomography similarly to $F$. 

We say that a state $\rho$ exhibits SWSSB if it is strongly symmetric ($U\rho \propto \rho$) and spontaneously breaks the stron symmetry but not the weak one. 
The strong symmetry is broken if the R\'enyi-1 correlator $R_1(i,j)$ remains finite in the limit of large spatial separation between the points $i$ and $j$.  
Similarly, the weak symmetry is unbroken if the standard correlator $C(i,j) = {\rm Tr}(\rho O_i O_j^\dagger)$ vanishes with increasing distance. 
We note that $R_1$ and $F$ are equivalent for this purpose~\cite{Liu_diagnosingstrongtoweaksymmetrybreaking_2024,weinstein_efficient_2024}. On the other hand $R_2$ does not enjoy a stability theorem with respect to strongly symmetric quantum channels~\cite{Lessa_strong_2025}, and so may not identify mixed state phases correctly. 

With this background, we can now give a definition of pseudo-SWSSB ensembles of mixed states.
\begin{definition}[Pseudo-SWSSB] \label{def:pseudoSWSSB}
    For a given symmetry group $G$, an ensemble of $N$-qubit density matrices $\mathcal{E}(G)$ exhibits ``pseudo-SWSSB" under $G$ if it satisfies the following properties:
    (i) each state $\rho \in \mathcal{E}(G)$ is efficiently preparable;
    (ii) each $\rho \in \mathcal{E}(G)$ is strongly symmetric under $G$;
    (iii) the strong $G$-symmetry is unbroken: $\sum_{i,j} R_1(i,j) = o(N^2)$ with high probability over $\rho\sim \mathcal{E}(G)$ as $N\to\infty$; 
    (iv) the ensemble $\mathcal{E}(G)$ is computationally indistinguishable from a state $\rho_g$ that exhibits SWSSB: for any number of copies $k={\rm poly}(N)$ and any efficient quantum algorithm $\mathcal{A}$,
    \begin{align}
        \left|\mathcal{A}\(\rho^{\otimes k}_g\)-\mathcal{A}\left({\mathbb{E}}_{\rho\sim \mathcal{E}(G)} \rho^{\otimes k} \right) \right|\leq o(1/{\rm poly}(N)).
    \end{align}
\end{definition}
Here a quantum algorithm is a binary measurement with outcomes $0,1$, i.e., $\mathcal{A}(\rho_g^{\otimes k}) = {\rm Tr}(E\rho_g^{\otimes k})$ with $0\leq E \leq \mathbb{I}$ an operator on $k$ Hilbert space copies. It is efficient if it can be implemented in ${\rm poly}(N)$ depth with ${\rm poly}(N)$ auxiliary qubits.

Let us remark on the meaning of Definition~\ref{def:pseudoSWSSB}. If there was an efficient algorithm $\mathcal{A}$ capable of diagnosing SWSSB given polynomially many copies of an unknown state $\rho$, then one could use $\mathcal{A}$ to distinguish the ensemble $\mathcal{E}(G)$ from $\rho_g^{\otimes k}$---which contradicts the definition of $\mathcal{E}(G)$. 
Thus the existence of a pseudo-SWSSB ensemble would rule out such an algorithm, establishing a fundamental limitation on the observability of SWSSB. 
Specifically, it would imply that protocols to detect SWSSB in arbitrary unknown states necessarily require a superpolynomial amount of resources (e.g. state copies). Conversely, efficient protocols to detect SWSSB necessarily require additional assumptions or prior information about the state $\rho$. A protocol in the latter category was proposed in Ref.~\cite{weinstein_efficient_2024}.

Another interesting observation on Definition~\ref{def:pseudoSWSSB} is that, in point (iii), we do not have to separately ask for the weak symmetry to be unbroken in $\mathcal{E}(G)$: this follows from property (iv) if we pick the efficient algorithm $\mathcal{A}$ that measures $O_i O_j^\dagger \otimes O_i^\dagger O_j$ on $k=2$ state copies, returning $|C(i,j)|^2$. Then, since $|C(i,j)|^2$ is negligibly small in $\rho_g$, the same must be true on average over $\mathcal{E}(G)$, bounding both mean and variance of $C(i,j)$ over the ensemble and showing that the weak symmetry is unbroken. 
This highlights a key distinction between strong and weak symmetry: 
it is impossible to hide SSB for the weak symmetry because the order parameter is a polynomial in $\rho$; conversely, since the order parameters for the strong symmetry are {\it not} polynomials in $\rho$, the possibility of concealing SSB exists.

In the following we present rigorous constructions of pseudo-SWSSB ensembles for $\mathbb{Z}_2$ symmetry and $U(1)$ symmetry, thus establishing the hardness of detecting SWSSB in these paradigmatic cases. 
Generalization to the less explored case of non-Abelian symmetry is an interesting direction for future work. 


{\it $\mathbb{Z}_2$ symmetry.}---We start by considering an $N$-qubit Hilbert space with a $\mathbb{Z}_2$ symmetry generated by $\bar{X}=\prod_iX_i$, with local charged operators $Z_i$ ($\bar{X} Z_i \bar{X} = -Z_i$). 
The prototypical SWSSB state in this case is $\rho_0=(\mathbb{I}+\bar{X})/d$, where $d=2^N$ is the Hilbert space dimension. This is the maximally mixed state in the subspace $\mathcal{H}_+$, the $+1$ eigenspace of $\bar{X}$. 
It is easy to verify that $\bar{X} \rho_0 = \rho_0$ (strong symmetry) and ${\rm Tr}(\rho_0 Z_i Z_j) = 0$ for all $i\neq j$ (unbroken weak symmetry). 
At the same time, since $Z_i Z_j \rho_0 Z_i Z_j = \rho_0$, we have $R_1(i,j) = F(i,j) = R_2(i,j) = 1$ for all $i$ and $j$ (broken strong symmetry). 

Before presenting our rigorous construction of the pseudo-SWSSB ensemble, let us qualitatively sketch the main ideas.
Consider a stochastic mixture of some number $r$ of independent Haar-random states inside symmetry sector $\mathcal{H}_+$: $\rho = r^{-1}\sum_{\alpha=1}^r \ketbra{\psi_\alpha}$, where $\bar{X} \ket{\psi_\alpha} = +\ket{\psi_\alpha}$ for all $\alpha$. Since any two states in the mixture are nearly orthogonal with high probability (if $r$ is not too big), this state has purity ${\rm Tr}(\rho^2) \approx 1/r$. The R\'enyi-2 correlator then reads $R_2(i,j) \simeq r^{-1} \sum_{\alpha,\beta} |\bra{\psi_\alpha} Z_i Z_j \ket{\psi_\beta}|^2$. Since the $\ket{\psi_\alpha}$ states are Haar-random, the typical value of each matrix element will be $|\langle \psi_\alpha | \psi_\beta \rangle |^2 = O(1/d)$, and $R_2(i,j) = O(r^2/d)$ for all $i\neq j$. If the chosen rank $r$ grows sub-exponentially in $N$, then $\rho$ does not break the strong symmetry, and so does not have SWSSB (at least according to $R_2$). Yet if $r$ is large enough, $\rho$ is a highly mixed state in a randomized basis of $\mathcal{H}_+$, and so we may expect it to be hard to distinguish from $\rho_0$, the maximally mixed state in the subspace, which does have SWSSB. 
Below we develop this idea into a derandomized construction that is efficiently implementable, and we prove absence of SWSSB according to $R_1$ as well as computational indistinguishability from the SWSSB state $\rho_0$. 

To derandomize our construction, it is essential to make use of {\it pseudorandom unitaries}:
\begin{definition}[Pseudorandom unitary~\cite{ji_pseudorandom_2018,metger_simple_2024,ma2024constructrandomunitaries,chen_pseudorandom_2024}]
    A family of ensembles $\{ \mathcal{U}_{{\rm PRU},N} \}_{N\in\mathbb{N}}$ of $N$-qubit unitaries is a PRU if
    (i) each $U \in \mathcal{U}_{{\rm PRU},N}$ is efficiently implemented in depth ${\rm poly}(N)$, 
    and 
    (ii) a unitary chosen uniformly at random from $\mathcal{U}_{{\rm PRU},N}$ is computationally indistinguishable from a Haar-random unitary: 
    any efficient algorithm querying $U$ in parallel $k = {\rm poly}(N)$ times gives the same average outcome whether $U$ is a PRU or a Haar-random unitary, up to error $o(1/{\rm poly}(N))$. 
\end{definition}
In particular we will use the recently introduced `PFC' ensemble~\cite{metger_simple_2024}, where each $U$ is the composition of a pseudorandom permutation of the computational basis~\cite{zhandry_note_2025}, a pseudorandom binary phase function~\cite{zhandry_how_2021}, and a random Clifford unitary~\cite{Berg_A_2021}. This has the advantage of additionally forming an exact 2-design~\cite{ambainis_quantum_2007,dankert_exact_2009,roberts_chaos_2017}, due to the random Clifford. 

We define our pseudo-SWSSB ensemble as
\begin{align}\label{eq:E2_def}
    \mathcal{E}(\mathbb{Z}_2):=\left\{\frac{U\Pi_rU^\dagger}{r}:\ U \sim \mathcal{U}_{\rm PRU}(\mathcal{H}_+)\right\},
\end{align}
with $\Pi_r$ a rank-$r$ projector supported inside the symmetry sector subspace $\mathcal{H}_+$ and $\mathcal{U}_{\rm PRU}(\mathcal{H}_+)$ a PRU ensemble on $\mathcal{H}_+$. 
These states can be prepared efficiently. A possible approach, illustrated in Fig.~\ref{fig:schematic}(a), is to take a product state $\ket{+}\otimes \ket{0}^{\otimes N-1}$ and apply a fully depolarizing channel to the last $\log_2(r)$ qubits (we can assume $r$ is a power of 2) to get a mixed state with the desired spectrum;
then we apply a PRU on qubits $i=2,\dots N$; finally we apply the encoder circuit for the repetition code where qubit $i = 1$ serves as the `logical' and qubits $i = 2,\dots N$ as the `syndromes'. The encoder circuit turns $X_1$ into the global symmetry operator $\bar{X}$, which ensures the output state is within the $\bar{X} = +1$ sector since the input has $X_1 = +1$. 
In all, this results in a state $U\Pi_r U^\dagger / r$ with $U$ a PRU on the symmetry sector subspace $\mathcal{H}_+$. The state can be prepared efficiently since the repetition code encoder is implementable as a depth-$N$ staircase of CNOT gates and PRUs are efficiently implementable by definition.

\begin{figure}
    \centering
    \includegraphics[width=0.99\columnwidth]{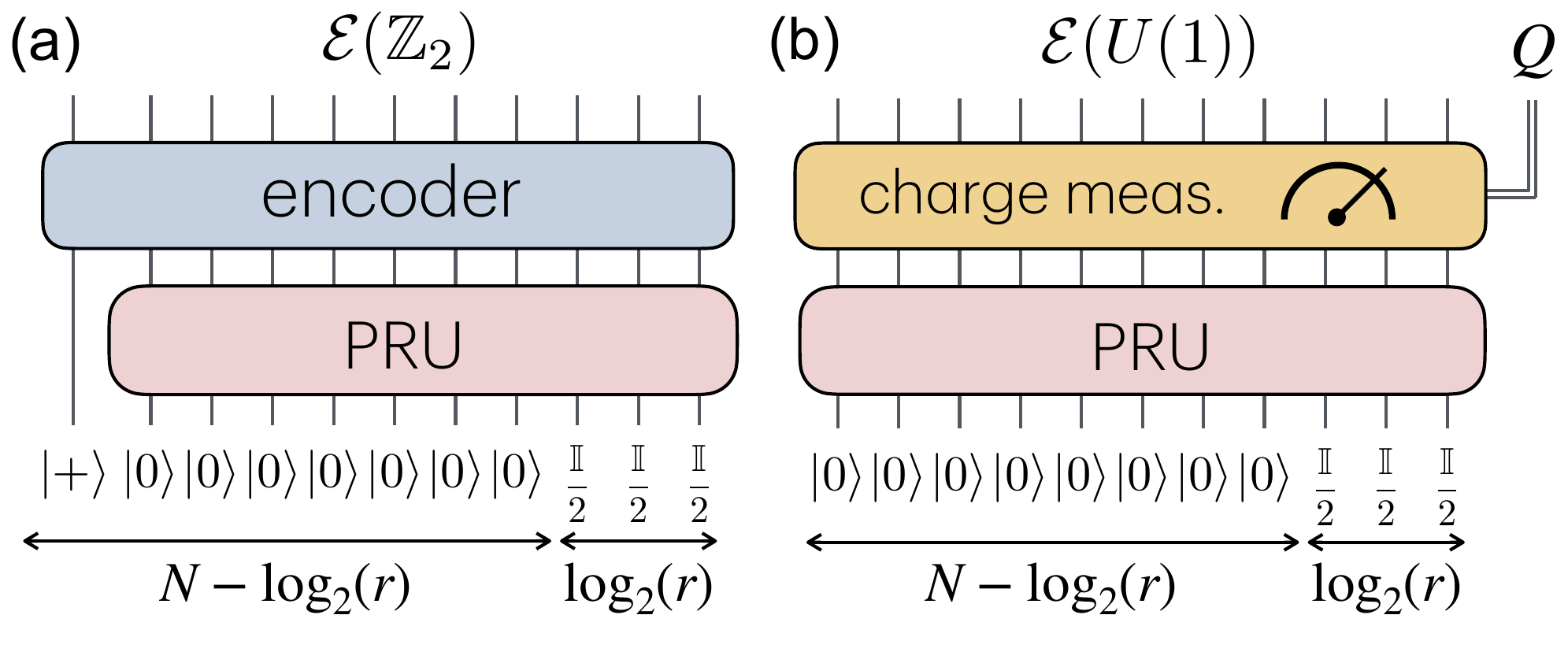}
    \caption{Circuits for the preparation of pseudo-SWSSB states with (a) $\mathbb{Z}_2$ symmetry or (b) $U(1)$ symmetry.
    In (a), the encoder is a staircase circuit of CNOT gates mapping $X_1$ to the symmetry generator $\bar{X} = \prod_i X_i$. 
    In (b), the charge measurement can be implemented with $O(N)$ ancilla qubits and $O(\log(N))$ depth; postselection of an outcome $Q$ is sample-efficient as long as $Q = N/2 + O(N^{1/2})$. 
    To achieve pseudo-SWSSB the number of depolarized qubits $\log_2(r)$ must be superlogarithmic but subextensive.}
    \label{fig:schematic}
\end{figure}

Next we prove that $\mathcal{E}(\mathbb{Z}_2)$ defined in Eq.~\eqref{eq:E2_def} has pseudo-SWSSB if $r$ scales suitably with $N$. 
Note that, since every $\rho\in \mathcal{E}(\mathbb{Z}_2)$ is proportional to a projector, we have $R_1(i,j)=R_2(i,j)$ and $R_1(i,j) \leq F(i,j)\leq\sqrt{R_1(i,j)}$~\cite{Liu_diagnosingstrongtoweaksymmetrybreaking_2024,weinstein_efficient_2024}, so all three diagnostics have the same asymptotic behavior in this case. We thus focus only on $R_1(i,j)$ in the following unless specified.

\begin{theorem}[Pseudo-SWSSB for $\mathbb{Z}_2$ symmetry] 
\label{thm:pseudo_swssb_Z2} 
For $r$ satisfying $\omega(\log N)<\log r<o(N)$, the ensemble $\mathcal{E}(\mathbb{Z}_2)$ in Eq.~\eqref{eq:E2_def} has pseudo-SWSSB.
\end{theorem}

\begin{proof}
    We prove the four conditions in Definition~\ref{def:pseudoSWSSB}. 
    (i) The ensemble is efficiently preparable as both the PRU and the encoder circuit have polynomial depth.
    (ii) Strong symmetry is ensured by the fact that each $\rho \in \mathcal{E}(\mathbb{Z}_2)$ is supported inside the $\mathcal{H}_+$ symmetry sector: $\bar{X}\rho = +\rho$.
    (iii) The strong symmetry is unbroken because the R\'enyi-1 correlator $R_1(i,j)$ is negligibly small with high probability. This holds because each state $\rho$ in the ensemble is proportional to a rank-$r$ projector, so $\sqrt{\rho} = \sqrt{r} \rho$, and $R_1(i,j) = r {\rm Tr}(\rho O_i O_j^\dagger \rho O_i^\dagger O_j)$. This is quadratic in $\rho = U\Pi_r U^\dagger/r$, and thus in the PRU $U$. Since the PFC ensemble contains a random Clifford unitary, it forms an exact 2-design. Thus the average of $R_1$ over the ensemble can be computed exactly; a straightforward calculation gives $\mathbb{E}_{\rho\sim \mathcal{E}(\mathbb{Z}_2)} R_1(i,j) \leq O(r / d)$. Since $R_1(i,j)$ is non-negative (as seen by expanding $\rho$ in its eigenbasis), Markov's inequality gives 
    \begin{equation}
        \textsf{Pr}(R_1(i,j) > \delta) < O\left(\frac{r}{\delta d}\right) \label{eq:thm1_r1_small}
    \end{equation} 
    for all $\delta > 0$. Taking e.g. $\delta = d^{-1/2}$ shows that, if $r$ is sub-exponential, then $R_1(i,j)$ is exponentially small with probability exponentially close to 1. Thus that the strong symmetry is unbroken.  
    (iv) Computational indistinguishability between $\rho_0$ and $\mathcal{E}(\mathbb{Z}_2)$ can be proven in two steps. First we replace the PRUs by genuine Haar-random unitaries, defining the ensemble 
        $\mathcal{E}'(\mathbb{Z}_2):=\left\{{U\Pi_rU^\dagger}/{r}:\ U \sim \mathcal{U}_{\rm Haar}(\mathcal{H}_+)\right\}$;
    this is computationally indistinguishable from $\mathcal{E}(\mathbb{Z}_2)$ by the definition of PRUs. 
    Then we prove the following (see Supplementary Information):
    \begin{lemma}
    \label{lem:stat_distance}
    For $e^{\omega(\log N)} < r < e^{o(N)}$ and any number of state copies $k={\rm poly}(N)$, the ensemble $\mathcal{E}'(\mathbb{Z}_2)$ satisfies
        \begin{align}
            \left\Vert {\mathbb{E}}_{\rho \sim \mathcal{E}'(\mathbb{Z}_2)}\rho^{\otimes k}-\rho_0^{\otimes k}\right\Vert_{\rm tr}\leq O\(\frac{k^2}{r}\). \label{eq:stat_indist_lemma}
        \end{align}
    \end{lemma}
    This asserts the statistical indistinguishability~\footnote{Statistical indistinguishability allows the adversary to use arbitrary algorithms (not necessarily efficient) on the given number of copies. It is a stronger condition than computational indistinguishability.} between the random ensemble $\mathcal{E}'(\mathbb{Z}_2)$ and $\rho_0$ for any number $k = {\rm poly}(N)$ of copies. 
    Thus $\mathcal{E}(\mathbb{Z}_2)$ is indistinguishable from $\mathcal{E}'(\mathbb{Z}_2)$ which is indistinguishable from $\rho_0$, giving property (iv) and concluding the proof. 
\end{proof}

Theorem~\ref{thm:pseudo_swssb_Z2} is one of the main results of this work. It proves the existence of efficiently preparable state ensembles that do not exhibit SWSSB but cannot be efficiently distinguished from the prototypical SWSSB state $\rho_0 = (\mathbb{I} + \bar{X})/d$. Thus mixed-state phases of matter based on $\mathbb{Z}_2$ SWSSB can be hard to detect without additional structure or prior information on the system. 
How hard exactly? The largest error term in our analysis, Eq.~\eqref{eq:stat_indist_lemma}, is $O(k^2/r)$. Thus, given a constant number $k$ of copies at a time (the most experimentally realistic scenario), it would take $\Omega(r^2)$ experimental repetitions~\footnote{If the trace distance between two states is $\epsilon$ it takes $\Omega(1/\epsilon^2)$ repetitions of the optimal measurement to distinguish them with high confidence.} to detect a discrepancy between $\mathcal{E}(\mathbb{Z}_2)$ and $\rho_0$ with high confidence. This bound is saturated by measuring the purity ${\rm Tr}(\rho^2) = 1/r$ from two-copy measurements via SWAP test. 
It follows that {\it no protocol} using constant quantum memory is asymptotically more efficient than a brute-force measurement of $R_2$. While $R_2$ itself is not a reliable measure of SWSSB, recent work \cite{zhang_fluctuation_2024} finds a series expansion of the fidelity correlator $F$ in terms of higher-R\'enyi correlators $R_{2n}$, suggesting that a sample complexity of ${\rm poly}(r)$ should be sufficient to estimate $F$, in agreement with our result. 


{\it $U(1)$ symmetry.}---We now generalize our finding from the discrete symmetry $\mathbb{Z}_2$ to the continuous symmetry $U(1)$. The latter acts on the Hilbert space as $e^{i\theta \sum_i (\mathbb{I}-Z_i)/2}$ for $e^{i\theta} \in U(1)$, splitting the $N$-qubit Hilbert space into a direct sum of charge sector subspaces $\mathcal{H}_Q$ of dimension $d_Q = \binom{N}{Q}$, each spanned by computational basis states of Hamming weight $Q$. The standard choice for local charged operators is $O_i = S_i^+ = (X_i + i Y_i)/2$ and $O_i^\dagger = S_i^-$. These are not unitary, so the fidelity correlator $F$ is ill-defined; we will focus on $R_1$ (one could alternatively take, e.g., $O_i = X_i$ without affecting our results).
In analogy with the $\mathbb{Z}_2$ case, the reference state for $U(1)$ SWSSB is a maximally mixed state inside a charge sector subspace $\mathcal{H}_Q$, $\rho_Q=P_Q/d_Q$, with $P_Q$ the projector on $\mathcal{H}_Q$.
$\rho_Q$ exhibits SWSSB if the charge density $\sigma:=\lim_{N\to\infty} Q/N$ is a finite constant not equal to 0 or 1 (see Supplementary Information). 

Adapting the ensemble construction in Eq.~\eqref{eq:E2_def} to the case of $U(1)$ symmetry is not straightforward because existing PRU ensembles require a Hilbert space with a qubit tensor product structure, while the charge sector we are interested in has dimension $\binom{N}{Q}$ which is not a power of 2 (or any one prime number). 
To avoid this issue, we follow a slightly different route sketched in Fig.~\ref{fig:schematic}(b). 
First we prepare a rank-$r$ mixed state $\Pi_r/r$ by depolarizing $\log_2(r)$ qubits in a pure product state; then we apply a PRU $U$ from the PFC ensemble to the whole $N$-qubit Hilbert space, obtaining a state $U\Pi_r U^\dagger / r$. 
At this point we perform a measurement of the charge $\frac{1}{2}\sum_i(Z_i-\mathbb{I})$, i.e., a Hamming weight projection, which can be done efficiently using $O(N)$ auxiliary qubits and $O(\log N)$ depth~\cite{Rethinasamy_Logarithmic_2024}. This measurement returns a value $Q$ of the charge distributed according to $p(Q) = d_Q/d = 2^{-N} \binom{N}{Q}$, a binomial distribution with mean $N/2$ and standard deviation $\sqrt{N}/2$. 
Thus, each value of $Q$ in the range $N/2 + O(\sqrt{N})$ is produced in the experiment with probability $\Omega(1/\sqrt{N})$, and can be postselected with polynomial sampling overhead. 
This protocol can be used to efficiently prepare the following ensemble for any value of $Q = N/2 + O(\sqrt{N})$:
\begin{align}\label{eq:E2_U1_def}
    \mathcal{E}(U(1))
    :=\left\{\frac{ P_QU\Pi_rU^\dagger P_Q}{\Tr\( P_QU\Pi_rU^\dagger\)}:\ 
    U\sim\mathcal{U}_{\rm PRU}(\mathcal{H})\right\}.
\end{align}

The second main result of our work is that the above ensemble exhibits $U(1)$-SWSSB:
\begin{theorem}[Pseudo-SWSSB for $U(1)$ symmetry] \label{thm:pseudo_swssb_U1}
     $\mathcal{E}(U(1))$ in Eq.~\eqref{eq:E2_U1_def} has pseudo-SWSSB if $e^{\omega(\log N)}<r<e^{o(N)}$ and $|Q-N/2|\leq O(N^{1/2})$.
\end{theorem}
The proof is given in the Supplementary Information. Here we provide a brief sketch. Properties (i) and (ii) (efficient preparation and strong symmetry) are apparent from the construction described above. Property (iv) (computational indistinguishability) is proven similarly to the $\mathbb{Z}_2$ case (Lemma~\ref{lem:stat_distance}), with some additional complications in the Weingarten calculus due to the presence of $U$ in the denominator in Eq.~\eqref{eq:E2_U1_def}.
Finally to establish property (iii) (absence of SWSSB) we use the fact that the PFC PRUs form an exact 2-design to show that $\rho\in\mathcal{E}(U(1))$ is extremely close to a (trace-normalized) projector with high probability, and thus reduce $R_1(i,j)$ to a quadratic function of the PRU $U$ with negligible error. At that point the methods used in the $\mathbb{Z}_2$ case carry over. 

Theorem~\ref{thm:pseudo_swssb_U1} shows that our statements on the hardness of observing SWSSB apply not just to discrete symmetries, but also to continuous ones. This is a nontrivial generalization due to the richer structure of symmetry sector subspaces and their lack of tensor product structure. Our approach, based on charge measurements and (sample-efficient) postselection, is straightforwardly generalizable to other Abelian symmetry groups, such as products of multiple $U(1)$ and $\mathbb{Z}_2$ groups.


{\it Discussion and outlook.}---We have shown that the phenomenon of strong-to-weak symmetry breaking, a signature of intrinsically mixed phases of matter, can be hard to detect experimentally. In particular we have proven that, given only standard experimental access to copies of an unknown mixed state, no efficient experiment can in general decide whether or not the state has SWSSB. Our approach leverages recent ideas from quantum information and quantum cryptography~\cite{ji_pseudorandom_2018,metger_simple_2024} to construct pseudorandom density matrices~\cite{Bansal_pseudorandom_2025} that can hide their lack of SWSSB from any efficient experiment. 
While we rule out efficient state-agnostic protocols, efficient detection may still be possible in the presence of additional structure or prior information on the system~\cite{weinstein_efficient_2024}. 

Our work leaves several open questions. First, our construction only addresses on-site Abelian symmetries. Generalization to more complex symmetries such as non-Abelian or higher-form symmetries is an interesting future direction. 
A limitation of our approach for $U(1)$ symmetry is that we cannot efficiently prepare states away from the largest symmetry sector subspaces, i.e., at asymptotic charge density $\sigma = \lim_{N\to\infty} Q/N \neq 1/2$. This is due to our measurement-based approach to project PRUs on the whole Hilbert space into a fixed symmetry sector. To relax this limitation, it would be interesting to develop PRU ensembles that can target subspaces of arbitrary dimension.

By establishing a fundamental limitation on the observability of intrinsically mixed phases of matter, our work adds to a growing body of research that leverages pseudorandomness to understand the hardness of observing key properties of quantum many-body systems such as entanglement~\cite{Aaronson_quantum_2024,giurgicatiron_pseudorandomnesssubsetstates_2023,jeronimo_pseudorandompseudoentangled_2024,cheng_pseudoentanglementtensornetworks_2024}, thermalization~\cite{feng_dynamics_2024,lee_fastpseudothermalization_2024}, chaos~\cite{gu_simulatingquantumchaoschaos_2024,lee2024pseudochaoticmanybodydynamicspseudorandom}, nonstabilizerness~\cite{Gu_pseudomagic_2024}, topological order~\cite{Schuster_randomunitariesextremelylow_2025}, and holography~\cite{cheng_pseudoentanglementtensornetworks_2024,engelhardt_spoofing_2024,akers_holographicpseudoentanglementcomplexityadscft_2024}. 
A better understanding of which features are or are not efficiently observable in quantum experiments remains an essential goal for future research. 

{\it Acknowledgments.} We thank Yuan Xue for helpful discussions. XF thanks Zhen Bi, Leonardo A. Lessa and Shengqi Sang for the useful discussions. XF was supported by a TQI Postdoctoral Fellowship.

\let\oldaddcontentsline\addcontentsline
\renewcommand{\addcontentsline}[3]{}
\bibliography{pseudo_swssb}
\let\addcontentsline\oldaddcontentsline


\clearpage
\widetext

\setcounter{equation}{0}
\setcounter{figure}{0}
\setcounter{table}{0}
\setcounter{page}{1}
\makeatletter
\renewcommand{\thesection}{S\arabic{section}}
\renewcommand{\theequation}{S\arabic{equation}}
\renewcommand{\thefigure}{S\arabic{figure}}

\begin{center}
\textbf{\large Supplementary Information: Hardness of observing strong-to-weak symmetry breaking} \\
\end{center}

\tableofcontents

\section{Review of useful facts about Weingarten calculus}

Here we review some facts about Weingarten calculus which are used in our proofs.

\begin{fact}[$k$-fold twirling channel] 
We have 
\label{fact:weingarten}
    \begin{align}
        \Phi^{(k)}_{\rm Haar}(O) 
        = \mathbb{E}_{U\sim {\rm Haar}} \[U^{\otimes k}OU^{\dagger\otimes k}\]=\sum_{\sigma,\tau\in S_k}\Wg_d(\sigma\tau^{-1})\Tr\(\hat{\tau}^{-1}O\)\hat{\sigma},
    \end{align}
where $\hat{\sigma}$ represents the action of permutation $\sigma \in S_k$ on the $k$ copies of the Hilbert space and $\textsf{Wg}_d$ is the Weingarten function on a $d$-dimensional Hilbert space. 
For $k = 2$, we have 
\begin{equation}
    \Wg_d(e) = \frac{1}{d^2-1},
    \qquad 
    \Wg_d(\pi) = -\frac{1}{d(d^2-1)},
\end{equation}
where $e,\pi\in S_2$ are the identity and the transposition respectively. 
\end{fact}

\begin{fact}[Bounds for absolute value of Weigarten function~\cite{collins_weingarten_2017,cheng_pseudoentanglementtensornetworks_2024}]
\label{fact:wg_abs_bound}
We have 
    \begin{align}
        |\Wg_d(\sigma)|\leq d^{-k}\(\frac{d}{4}\)^{|\sigma|-k}\[1+O\(\frac{k^{7/2}}{d^2}\)\]
    \end{align}
    where $|\sigma|$ is the number of cylces in permutation $\sigma$. 
    Additionally, for $\sigma = e$ (identity permutation), we have~\cite{Zhou_emergent_2019}
    \begin{align}
        \Wg_d(e)=d^{-k}\[1+O\(\frac{k^2}{d^2}\)\].
    \end{align}
\end{fact}

\begin{fact}[Gram matrix sum]\label{fact:trace_gram} We have
    \begin{align}
        \sum_{\sigma\in S_k}d^{|\sigma|}=\frac{(d-1+k)!}{(d-1)!}=d^k\[1+O\(\frac{k^2}{d}\)\].
    \end{align}
\end{fact}

\section{Proof of Lemma~\ref{lem:stat_distance}}\label{ssec:proof_lem3}

Here we will prove the statistical indistinguishability between the random ensemble $\mathcal{E}'(\mathbb{Z}_2)$ and the state $\rho_0 = (\mathbb{I} + \bar{X}) / d = (2/d) P_0 $. We take a number of state copies $k={\rm poly}(N)$ and focus on the regime $e^{\omega(\log N)}<r<e^{o(N)}$ for the parameter $r$ that sets the rank of states in the ensemble.

We first use Fact~\ref{fact:weingarten} to rewrite $\mathbb{E}_{\rho \sim \mathcal{E}'(\mathbb{Z}_2)} \rho^{\otimes k} = \Phi^{(k)}_{\rm Haar}(\Pi_r^{\otimes k} / r^k) = \sum_{\sigma,\tau\in S_k}\textsf{Wg}_{d/2}(\sigma\tau^{-1})\frac{r^{|\tau|}}{r^k}\hat{\sigma}$. Recall that the Haar measure is over the subspace $\bar{X} = +1$, of dimension $d/2$. 
Then, the trace distance can be bounded as
\begin{equation}\label{seq:tr_dis}
    \tracenorm{\sum_{\sigma,\tau \in S_k }\textsf{Wg}_{d/2}(\sigma\tau^{-1})\frac{r^{|\tau|}}{r^k}\hat{\sigma}-\frac{P_0^{\otimes k}}{(d/2)^{k}}}
    \le \left|\sum_{\tau \in S_k }\textsf{Wg}_{d/2}(\tau^{-1})\frac{r^{|\tau|}}{r^k}\(\frac{d}{2}\)^k-1\right| 
    + \tracenorm{\sum_{\substack{ \sigma, \tau\in S_k:\\ \sigma \ne e}}\textsf{Wg}_{d/2}(\sigma\tau^{-1})\frac{r^{|\tau|}}{r^k}\hat{\sigma}},
\end{equation}
where we used the fact that $\hat{e}=P_0^{\otimes k}$ and $\Vert P_0\Vert_{\rm tr}=d/2$. 

The first term of the right hand side of Eq.\eqref{seq:tr_dis} can be further bounded by the sum of $\textsf{Wg}_{d/2}(e)=(2/d)^k[1+O(k^2/d^2)]$ (Fact~\ref{fact:wg_abs_bound}) and
\begin{align}
    \left|\sum_{\tau\in S_k:\ \tau \ne e}\textsf{Wg}_{d/2}\(\tau^{-1}\)\frac{r^{|\tau|}(d/2)^k}{r^k}\right|\le 
    &-1+ \sum_{\tau\in S_k }\left|\textsf{Wg}_{d/2}\(\tau^{-1}\)\right|\frac{r^{|\tau|}(d/2)^k}{r^k} \nonumber\\
    \le& -1 + \sum_{\tau\in S_k }\left(\frac{dr}{8}\right)^{|\tau|-k}\[1+O\(\frac{k^{7/2}}{d^2}\)\]
    =O\left(\frac{k^2}{rd}\right), \label{eq:si_triangle}
\end{align}
where we used Facts~\ref{fact:wg_abs_bound} and~\ref{fact:trace_gram}.
For the second term on the right hand side of Eq.\eqref{seq:tr_dis}, we can use the inequaility $\tracenorm{A}\le\textsf{dim} \cdot \opnorm{A}$ with $\textsf{dim}$ the Hilbert space dimension, and Facts~\ref{fact:wg_abs_bound} and~\ref{fact:trace_gram}, to get
\begin{align}
    \tracenorm{\sum_{\substack{\sigma,\tau\in S_k:\\ \sigma \ne e}}\textsf{Wg}_{d/2}\(\sigma\tau^{-1}\)\frac{r^{|\tau|}}{r^k}\hat{\sigma}}
    \le &\opnorm{ \sum_{\substack{\sigma,\tau\in S_k:\\ \sigma \ne e}} \textsf{Wg}_{d/2}\(\sigma\tau^{-1}\)\frac{r^{|\tau|}\(d/2\)^k}{r^k}\hat{\sigma}} \nonumber\\
    \le&-1 + \sum_{\sigma,\tau \in S_k }\left|\textsf{Wg}_{d/2}\(\sigma\tau^{-1}\)\right|\frac{r^{|\tau|}(d/2)^k}{r^k} \nonumber\\
    \leq& -1 + \sum_{\sigma,\tau \in S_k }\(\frac{d}{8}\)^{\left|\sigma\tau^{-1}\right|-k}r^{|\tau|-k}\[1+O\(\frac{k^{7/2}}{d^2}\)\]
    =O\(\frac{k^2}{r}\).
\end{align}
Plugging the two bounds back in Eq.~\eqref{seq:tr_dis}, we obtain the statement of Lemma~\ref{lem:stat_distance}:
\begin{equation}
    \tracenorm{\mathbb{E}_{\rho\sim \mathcal{E}'(\mathbb{Z}_2) } \rho^{\otimes k} - \rho^{\otimes k}}\le O\left(\frac{k^2}{r}\right).
\end{equation}


\section{Proof of Theorem~\ref{thm:pseudo_swssb_U1}}

Here we prove Theorem~\ref{thm:pseudo_swssb_U1}, asserting pseudo-SWSSB in the ensemble $\mathcal{E}(U(1))$ of Eq.~\eqref{eq:E2_U1_def}. 
Let us briefly summarize some notation. The $U(1)$ symmetry action is $U_g(\theta)=e^{i\theta\hat{Q}}$ with charge operator $\hat{Q}=\sum_i(\mathbb{I}-Z_i)/2$. The Hilbert space is the direct sum of symmetry sectors, $\mathcal{H} = \bigoplus_{Q=0}^N \mathcal{H}_Q$. Each $\mathcal{H}_Q$ is spanned by computational basis states of Hamming weight $Q$ and has dimension $d_Q = \binom{N}{Q}$. To diagnose SWSSB we choose local charged operators $O_i=S_i^+ = (X_i + i Y_i) / 2$. These are not unitary, so the fidelity correlator $F$ is not defined; we use $R_1$ instead. 

First, we show that the states $\rho_Q=P_Q/ d_Q$ (maximally mixed states in a charge sector) exhibit SWSSB when the charge density $\sigma=Q/N$ is asymptotically not 0 or 1. 
Clearly $\rho_Q$ has a strong symmetry: $U_g(\theta) \rho_Q = e^{i\theta Q} \rho_Q$. 
Also, the ordinary correlator $C(i,j) = \Tr(S_i^+S_j^- P_Q/ d_Q)=0$ vanishes by the definition of $ P_Q$ (as seen e.g. by taking the trace in the computational basis), so the weak symmetry is unbroken. 
However, the the R\'enyi-1 correlator is given by
\begin{align}
    R_1(i,j) = \frac{1}{ d_Q}\Tr( P_QS_i^+S_j^- P_QS_j^+S_i^-)
    = \frac{1}{d_Q} \binom{N-2}{Q-1} 
    = \frac{Q(N-Q)}{N(N-1)}
    = \sigma(1-\sigma)+o(1),
\end{align}
where $\binom{N-2}{Q-1}$ is the number of computational basis states with total Hamming weight $Q$, a $|0\rangle$ at site $i$, and a $|1\rangle$ at site $j$. This is a finite constant as long as $\sigma\neq 0,1$, showing spontaneous breaking of the strong symmetry. 

Next, we show that the pseudorandom state ensemble 
\begin{align}
    \mathcal{E}(U(1)) = \left\{\frac{ P_QU\Pi_rU^\dagger P_Q}{\Tr\( P_QU\Pi_rU^\dagger P_Q\)}:\  U\sim \mathcal{U}_{\rm PRU}(\mathcal H) \right\} 
\end{align}
does not break the strong symmetry; see Lemma~\ref{nlem:renyi_corr_small}.

Finally, to show computational indistinguishability of $\mathcal{E}(U(1))$ from $\rho_Q$, we first use the definition of PRUs to replace $\mathcal{E}(U(1))$ by its Haar-random counterpart 
\begin{align}
    \mathcal{E}' (U(1)) = \left\{\frac{ P_QU\Pi_rU^\dagger P_Q}{\Tr\( P_QU\Pi_rU^\dagger P_Q\)}:\  U\sim \mathcal{U}_{\rm Haar} (\mathcal{H}) \right\}
\end{align}
in a computationally undetectable way; then we show that $\mathcal{E}' (U(1))$ is statistically indistinguishable from $\rho_Q$, see Lemma~\ref{nlem:stat_distance}. 
This concludes the proof of pseudo-SWSSB in $\mathcal{E}(U(1))$.

\begin{lemma}[absence of SWSSB in $\mathcal{E}(U(1))$]  \label{nlem:renyi_corr_small}
    The R\'enyi-1 correlator is superpolynomially small with high probability over the pseudorandom ensemble $\mathcal{E}(U(1))$: for all $\delta>0$, we have
    \begin{equation}
        {\sf Prob} \left( |R_1 (i,j)| \geq \delta \right) \leq O \left( \frac{r^{3/2}}{\delta d_Q^{1/4}} \right). 
    \end{equation}
\end{lemma} 
\begin{proof}
We will first show that the states $\rho \in \mathcal{E}(U(1))$ are extremely close to rank-$r$ projectors (up to the trace normalization) with high probability; that will let us reduce the expression for $R_1$ to a quadratic polynomial in $U,U^\dagger$, where we can use the 2-design property\cite{ambainis_quantum_2007,dankert_exact_2009,roberts_chaos_2017} to switch from the pseudorandom ensemble $\mathcal{U}_{\rm PRU}$ to the Haar-random ensemble and conclude the proof.

Let us first define the operators $\tilde{\rho} = \frac{d}{d_Q r} P_Q U \Pi_r U^\dagger P_Q$. These are un-normalized ``states'' of rank $\leq r$ ($\Pi_r$ has rank $r$ by definition, and conjugation by $P_Q$ cannot increase the rank). 
It is straightforward to show, using the 2-design property of $\mathcal{U}_{\rm PRU}$ and Weingarten calculus, that 
\begin{equation}
    \mathbb{E}_{U \sim \mathcal{U}_{\rm PRU}} {\rm Tr}(\tilde{\rho}) = 1,
    \quad 
    \mathbb{E}_{U \sim \mathcal{U}_{\rm PRU}} {\rm Tr}(\tilde{\rho}^2) = \frac{1}{r} + O\left( \frac{1}{d_Q}\right).  
    \label{eq:tr_rhotilde_avg}
\end{equation}
Thus the $\tilde{\rho}$ ``states'' are normalized on average and nearly maximally mixed on their rank-$r$ support. 
Specifically, let us write the (nonzero) eigenvalues of $\tilde{\rho}$ as $\{ 1/r +  \delta_i\}_{i=1}^r$; pluggint these into Eq.~\eqref{eq:tr_rhotilde_avg} we see that the $\delta_i$ satisfy
\begin{equation}
    \sum_i \mathbb{E}_{U \sim \mathcal{U}_{\rm PRU}} [\delta_i] = 0,
    \quad 
    \sum_i \mathbb{E}_{U \sim \mathcal{U}_{\rm PRU}} [\delta_i^2] = O(1/d_Q).
\end{equation}
By Markov's inequality, we have $\textsf{Pr}\left((\sum_i\delta_i^2) > c\right) \leq O(1/(cd_Q))$, and choosing e.g. $c = 1/d_Q^{1/2}$ we have that $\sum_i\delta_i^2\leq O\(1/d_Q^{1/2}\)$ with probability at least $1-O\(1/d_Q^{1/2}\)$.

As a consequence, when $\sum_i\delta_i^2\leq O(1/d_Q^{1/2})$, $\sqrt{\rho}$ is extremely close to $\sqrt{r} \tilde{\rho}$: recalling that $\rho = \tilde{\rho} / {\rm Tr}(\tilde \rho)$, we have
\begin{align}
    \left\Vert \sqrt{\rho} - \sqrt{r} \tilde{\rho} \right\Vert_{\rm tr}
    & = \sum_{i=1}^r \left|\left( {\frac{1 + r\delta_i}{\sum_{j=1}^r (1+r\delta_j)}}\right)^{1/2}  - \frac{1+r\delta_i}{\sqrt{r}} \right| \nonumber \\
    & \leq  \sum_{i=1}^r \frac{1}{\sqrt r} \left(\left| \left( {\frac{1 + r\delta_i}{1+\sum_{j=1}^r \delta_j}}\right)^{1/2}  - 1\right| + r|\delta_i|\right) \nonumber \\
    & \leq 2\sqrt{r} \sum_{i=1}^r |\delta_i|\leq2r \left( \sum_i\delta_i^2 \right)^{1/2}
    \leq O\left(\frac{r}{d_Q^{1/4}}\right).\label{eq:bound_sqrtrho}
\end{align} 
where we have used $\sum_i \delta_i \leq \sum_i |\delta_i| \leq \sqrt{r \sum_i\delta_i^2}$ (by Cauchy-Schwarz). 
Thus the R\'enyi-1 correlator simplifies: letting $B = O_i O_j^\dagger$, we have
\begin{align}
    R_1(i,j) 
    & = {\rm Tr}(\sqrt{\rho} B \sqrt{\rho} B^\dagger ) \nonumber \\
    & = \sqrt{r} {\rm Tr}(\tilde{\rho}B \sqrt{\rho} B) + {\rm Tr}((\sqrt{\rho} - \sqrt{r} \tilde{\rho}) B \sqrt{\rho} B^\dagger ) 
    \nonumber \\
    & \leq \sqrt{r} {\rm Tr}(\tilde{\rho}B \sqrt{\rho} B)  + \| \sqrt{\rho} - \sqrt{r} \tilde{\rho}\|_{\rm tr} \|  B \sqrt{\rho} B^\dagger \|_{\rm op} \nonumber \\ 
    & \leq \sqrt{r} {\rm Tr}(\tilde{\rho}B \sqrt{\rho} B)  + O\left( r/d_Q^{1/4} \right),
\end{align}
where we used the inequality $\| AB\|_{\rm tr} \leq \|A\|_{\rm tr} \|B\|_{\rm op}$, the bound from Eq.~\eqref{eq:bound_sqrtrho}, and the fact that $\| B\sqrt{\rho}B^\dagger\|_{\rm op} \leq \|B\|_{\rm op}^2 \|\rho\|_{\rm op}^{1/2} \leq 1$. 
Iterating the same argument on the remaining $\sqrt{\rho}$ operator, one gets
\begin{align}
    R_1(i,j) 
    & \leq r {\rm Tr}(\tilde{\rho} B \tilde\rho B^\dagger) 
    + \sqrt{r} \| \sqrt{\rho} - \sqrt{r} \tilde{\rho}\|_{\rm tr} \|B\rho B^\dagger \|_{\rm op} {\rm Tr}(\tilde{\rho})  
    + O\left( r/d_Q^{1/4}\right)  \nonumber \\
    & \leq r {\rm Tr}(\tilde{\rho} B \tilde\rho B^\dagger) 
    + {\rm Tr}(\tilde\rho) O\left( r^{3/2}/d_Q^{1/4}\right), 
\end{align}
where we also used $\tilde\rho = \rho {\rm Tr}(\tilde\rho)$. 
Finally we can average over ${U \sim \mathcal{U}_{\rm PRU}}$, breaking down the average between elements where $\sum_i \delta_i^2 < c/d_Q^{1/2}$ and the rest. Calling the first event $\mathcal{A}$ and its complement $\bar{\mathcal A}$, we have
\begin{align}
    \mathbb{E}_{U \sim \mathcal{U}_{\rm PRU}}     R_1(i,j)
    & = \textsf{Pr}(\mathcal{A})\mathbb{E}_{U\sim\mathcal{U}_{PRU}|\mathcal A }\left[R_1(i,j)\right]
    +\textsf{Pr}(\bar{\mathcal{A}})\mathbb{E}_{U\sim \mathcal{U}_{PRU}|\bar{\mathcal A}} \left[R_1(i,j)\right]\nonumber\\
    & \leq r \mathbb{E}_{U \sim \mathcal{U}_{\rm PRU}}[{\rm Tr}(\tilde{\rho} O_i O_j^\dagger \tilde{\rho} O_i^\dagger O_j)] + O\left( \frac{r^{3/2}}{d_Q^{1/4}} \right)+P(\bar{\mathcal{A}})\nonumber\\
    &\le \mathbb{E}_{U \sim \mathcal{U}_{\rm PRU}}[{\rm Tr}(\tilde{\rho} O_i O_j^\dagger \tilde{\rho} O_i^\dagger O_j)]+ O\left( \frac{r^{3/2}}{d_Q^{1/4}} \right)+O\left(\frac{1}{d_Q^{1/2}}\right) , 
    \label{eq:bounding_r1_pru}
\end{align}
Here we have used the facts that $\textsf{Pr}(\mathcal{A})\mathbb{E}_{|\mathcal{A}} [R_1]\le \mathbb{E}[R_1]$ since $R_1$ is always non-negative and $R_1\le 1$, along with the fact that $\textsf{Pr}(\bar{\mathcal A}) \leq O(1/d_Q^{1/2})$.
At this point, since the quantity to average is quadratic in $U$ and $\mathcal{U}_{\rm PRU}$ forms an exact 2-design, we may replace the PRU ensemble by the Haar ensemble:
\begin{align}
    \mathbb{E}_{U \sim \mathcal{U}_{\rm PRU}}     R_1(i,j) 
    & \leq r \frac{d^2}{d_Q^2} \mathbb{E}_{U \sim \mathcal{U}_{\rm Haar}}\left[{\rm Tr}\left(P_Q U \frac{\Pi_r}{r} U^\dagger P_Q O_i O_j^\dagger P_Q U \frac{\Pi_r}{r} U^\dagger P_Q O_i^\dagger O_j\right)\right] + O\left( \frac{r^{3/2}}{d_Q^{1/4}} \right) \nonumber \\
    & \leq r \frac{d^2}{d_Q^2} {\rm Tr}\left[ \Phi^{(2)}_{\rm Haar} \left( \frac{\Pi_r^{\otimes 2}}{r^2} \right) (P_Q O_i O_j^\dagger P_Q)\otimes(P_Q O_i^\dagger O_j P_Q) \hat\pi \right] + O\left( \frac{r^{3/2}}{d_Q^{1/4}} \right) \nonumber \\
    & \leq \frac{r}{d_Q^2} \frac{d^2}{d^2-1} \left( 1-\frac{1}{rd}\right) {\rm Tr}(P_Q O_i^\dagger O_j P_Q O_i O_j^\dagger ) + O\left( \frac{r^{3/2}}{d_Q^{1/4}} \right) \nonumber \\
    & \leq \frac{r}{d_Q} + O\left( \frac{r^{3/2}}{d_Q^{1/4}} \right) \leq O\left( \frac{r^{3/2}}{d_Q^{1/4}} \right).
\end{align}
Here we used Fact~\ref{fact:weingarten} to expand the Haar twirling channel. We then used the fact that ${\rm Tr}(P_Q O_i O_j^\dagger) = 0$ and that ${\rm Tr}(P_Q O_i O_j^\dagger P_Q O_i^\dagger O_j) \leq d_Q$, as seen e.g. by taking the trace in the computational basis. 
Finally, Markov's inequality yields the result. 
\end{proof}

\begin{lemma}[Statistical indistinguishability between $\rho_Q$ and $\mathcal{E}'(U(1))$] \label{nlem:stat_distance}
    \begin{align}
        \left\Vert\mathbb{E}_{\rho\sim \mathcal{E}'(U(1))} \rho^{\otimes k} -  \rho_Q^{\otimes k} \right\Vert_{\rm tr}\leq O\(\frac{k^{2}}{r}\).
    \end{align}
\end{lemma}
\begin{proof}
    We again introduce the unnormalized density matrix $\Tilde{\rho}:= \frac{d}{rd_Q} P_QU\Pi_rU^\dagger P_Q$ and recall that $\rho_Q = P_Q / d_Q$. 
    Since $\rho = \Tilde{\rho} / {\rm Tr}(\Tilde{\rho})$, by the triangle inequality we have 
    \begin{align}
        \left\Vert\mathbb{E}_U\frac{\Tilde{\rho}^{\otimes k}}{\(\Tr\Tilde{\rho}\)^k}-\frac{ P_Q^{\otimes k}}{ d_Q^k}\right\Vert_{\rm tr}\leq&\left\Vert\mathbb{E}_U\frac{\Tilde{\rho}^{\otimes k}}{\(\Tr\Tilde{\rho}\)^k}-\frac{\mathbb{E}_U\Tilde{\rho}^{\otimes k}}{\mathbb{E}_U\(\Tr\Tilde{\rho}\)^k}\right\Vert_{\rm tr} 
        + \left\Vert\frac{\mathbb{E}_U\Tilde{\rho}^{\otimes k}}{\mathbb{E}_U\(\Tr\Tilde{\rho}\)^k}-\frac{ P_Q^{\otimes k}}{ d_Q^k}\right\Vert_{\rm tr}. \label{eq:tracedistance_decomposition}
    \end{align}
    where $\mathbb{E}_{U}$ stands for $\mathbb{E}_{U\sim \mathcal{U}_{\rm Haar}}$ in the following.

    For the first term in the right hand side of Eq.~\eqref{eq:tracedistance_decomposition}, we have
    \begin{align}
        \left\Vert\mathbb{E}_U\frac{\Tilde{\rho}^{\otimes k}}{\(\Tr\Tilde{\rho}\)^k}-\frac{\mathbb{E}_U\Tilde{\rho}^{\otimes k}}{\mathbb{E}_U\(\Tr\Tilde{\rho}\)^k}\right\Vert_{\rm tr} 
        & \leq \left\Vert \mathbb{E}_U \left[ \rho^{\otimes k} \left( 1 - \frac{( {\rm Tr}\tilde\rho)^k}{\mathbb{E}_{U'} ({\rm Tr}\tilde\rho)^k} \right) \right] \right\Vert_{\rm tr}  \nonumber \\
        & \leq
        \mathbb{E}_U\left|1-\frac{\(\Tr\Tilde{\rho}\)^{k}}{\mathbb{E}_{U'} \(\Tr\Tilde{\rho}\)^k}\right|\nonumber\\
        & \leq \sqrt{\frac{\mathbb{E}_U\(\Tr\Tilde{\rho}\)^{2k}}{\(\mathbb{E}_U\(\Tr\Tilde{\rho}\)^k\)^2}-1},
    \end{align}
    where we have used Jensen's inequality for the function $\sqrt{x}$.
    Now we define $f(k):=\mathbb{E}_U\( \Tr \Tilde{\rho} \)^{k}$, 
    so that the above bound is given by $\sqrt{f(2k)/f(k)^2-1}$. Using Fact~\ref{fact:weingarten}, we have
    \begin{align}
        f(k)
        & = \left( \frac{d}{rd_Q} \right)^k \sum_{\sigma,\tau\in S_k}\Wg_d(\sigma\tau^{-1})r^{|\tau|} d_Q^{|\sigma|}\nonumber\\
        & = 1+O\(\frac{k^2}{d^2}\)+ \sum_{\substack{\sigma,\tau\in S_k: \\ (\sigma,\tau)\neq (e,e)} } d^k\Wg_{d}(\sigma\tau^{-1})r^{|\tau|-k} d_Q^{|\sigma|-k}.
    \end{align}
    We then bound the summation over $(\sigma, \tau)\neq (e,e)$ by
    \begin{align}
        \left|\sum_{\substack{\sigma,\tau\in S_k: \\ (\sigma,\tau)\neq (e,e)} }d^k\Wg_{d}(\sigma\tau^{-1})r^{|\tau|-k} d_Q^{|\sigma|-k}\right|
        & \leq
        \sum_{\sigma,\tau \in S_k }d^k\left|\Wg_{d}(\sigma\tau^{-1})\right|r^{|\tau|-k} d_Q^{|\sigma|-k}-1\nonumber\\
        & \leq \sum_{\sigma,\tau \in S_k }\(\frac{d}{4}\)^{|\sigma\tau^{-1}|-k}r^{|\tau|-k} d_Q^{|\sigma|-k}\[1+O\(\frac{k^{7/2}}{d^2}\)\]-1\nonumber\\
        & \leq \sum_{\sigma,\tau \in S_k }\(\frac{d}{4}\)^{|\sigma\tau^{-1}|-k} d_Q^{|\sigma|-k}\[1+O\(\frac{k^{7/2}}{d^2}\)\]-1\nonumber\\
        & \leq O\(\frac{k^{2}}{ d_Q}\),
    \end{align}
    where we have used Facts~\ref{fact:wg_abs_bound} and~\ref{fact:trace_gram}. 
    In all, we have $f(k)=1+O\(\frac{k^{2}}{ d_Q}\)$ and thus
    \begin{align}
        \sqrt{\frac{f(2k)}{f(k)^2}-1}=O\(\frac{k^{2}}{ d_Q}\),
    \end{align}
    which provides a negligible upper bound on the first term in Eq.~\eqref{eq:tracedistance_decomposition}. 

    Moving on to the second term in Eq.~\eqref{eq:tracedistance_decomposition}, we have
    \begin{align}
        \left\Vert\frac{\mathbb{E}_U\Tilde{\rho}^{\otimes k}}{\mathbb{E}_U\(\Tr\Tilde{\rho}\)^k}-\frac{ P_Q^{\otimes k}}{ d_Q^k}\right\Vert_{\rm tr}
        & = \left\Vert\frac{\sum_{\sigma,\tau\in S_k}d^k \Wg_d(\sigma\tau^{-1})r^{|\tau|-k} P_Q^{\otimes k} \hat{\sigma} P_Q^{\otimes k}}{d_Q^k f(k)}-\frac{ P_Q^{\otimes k}}{ d_Q^k}\right\Vert_{\rm tr}\nonumber\\
        &\leq  \left|\frac{ \sum_{\tau \in S_k }d^k \Wg_d\(\tau\)r^{|\tau|-k}}{f(k)}-1\right|+\left\Vert\frac{\sum_{\sigma, \tau\in S_k: \ \sigma\neq e}d^k \Wg_d(\sigma\tau^{-1})r^{|\tau|-k} P_Q^{\otimes k} \hat{\sigma} P_Q^{\otimes k}}{d_Q^k f(k)}\right\Vert_{\rm tr},
    \end{align}
    where we have separated the $\sigma=e$ term from the $\sigma\neq e$ terms. The former is bounded by
    \begin{align}
        \left|\frac{\sum_{\tau\in S_k }d^k \Wg_d\(\tau\)r^{|\tau|-k}}{f(k)}-1\right|
        & \leq \left|\frac{ \sum_{\tau\neq e} d^k \Wg_d\(\tau\)r^{|\tau|-k}}{f(k)}\right|+\left|\frac{ d^k\Wg_d\(e\)}{f(k)}-1\right|\nonumber\\
        & \leq \frac{ d^k\sum_{\tau\neq e}\left|\Wg_d\(\tau\)\right|r^{|\tau|-k}}{f(k)}+O\(\frac{k^{2}}{ d_Q}\)\nonumber\\
        & \leq \frac{1}{f(k)}\sum_{\tau\neq e}\(\frac{rd}{4}\)^{|\tau|-k}\[1+O\(\frac{k^{7/2}}{d^2}\)\]+O\(\frac{k^{2}}{ d_Q}\)\nonumber\\
        & \leq O\(\frac{k^{2}}{ d_Q}\),
    \end{align}
    where we have used Facts~\ref{fact:wg_abs_bound} and~\ref{fact:trace_gram};
    the latter is bounded by
    \begin{align}
        \left\Vert\frac{1}{d_Q^k f(k)} \sum_{\substack{ \sigma,\tau \in S_k:\\ \sigma \neq e}} d^k \Wg_d(\sigma\tau^{-1})r^{|\tau|-k} P_Q^{\otimes k} \hat{\sigma} P_Q^{\otimes k} \right\Vert_{\rm tr}
        & \leq \frac{d^k}{d_Q^k f(k)} \sum_{\substack{ \sigma,\tau \in S_k:\\ \sigma \neq e}} \left|\Wg_d(\sigma\tau^{-1})\right|r^{|\tau|-k}\left\Vert P_Q^{\otimes k} \hat{\sigma} P_Q^{\otimes k}\right\Vert_{\rm tr}\nonumber\\
        & \leq \frac{d^k}{d_Q^k f(k)} \sum_{\substack{ \sigma,\tau \in S_k:\\ \sigma \neq e}} \left|\Wg_d(\sigma\tau^{-1})\right|r^{|\tau|-k}\left\Vert P_Q^{\otimes k} \hat{\sigma} P_Q^{\otimes k}\right\Vert_{\rm op} d_Q^k \nonumber\\
        & \leq \frac{d^k}{f(k)} \sum_{\substack{ \sigma,\tau \in S_k:\\ \sigma \neq e}} \left|\Wg_d(\sigma\tau^{-1})\right|r^{|\tau|-k} \nonumber\\
        & \leq \frac{1}{f(k)}\[ \sum_{\substack{ \sigma,\tau \in S_k:\\ \sigma \neq e}} \(\frac{d}{4}\)^{|\sigma\tau^{-1}|-k}r^{|\tau|-k}\]\[1+O\(\frac{k^{7/2}}{d^2}\)\]\nonumber\\
        & \leq O\(\frac{k^{2}}{r}\),
    \end{align}
    where we used Facts~\ref{fact:wg_abs_bound} and~\ref{fact:trace_gram} and inequality $\Vert AB\Vert_{\rm tr}\leq\Vert A\Vert_{\rm op}\Vert B\Vert_{\rm tr}$. 
    Combining all above contributions, we find that the leading error term is $O(k^2/r)$, thus completing the proof. 
\end{proof}

\end{document}